\documentclass[10pt,journal,compsoc]{IEEEtran}

\ifCLASSOPTIONcompsoc
  \usepackage[nocompress]{cite}
\else
  \usepackage{cite}
\fi

\ifCLASSINFOpdf
\else
\fi
\usepackage[utf8]{inputenc} 
\usepackage[T1]{fontenc}    
\usepackage{hyperref}       
\usepackage{url}            
\usepackage{booktabs}       
\usepackage{amsfonts}       
\usepackage{nicefrac}       
\usepackage{microtype}      
\usepackage{amsmath}
\usepackage{graphicx}
\usepackage{xcolor}
\usepackage{colortbl}
\usepackage{multirow}
\usepackage{float}
\usepackage{subcaption}

\usepackage{enumitem}
\setlist[itemize]{leftmargin=*}
\setlist[enumerate]{leftmargin=*}
\usepackage{array} 
\newcolumntype{P}[1]{>{\centering\arraybackslash}p{#1}}
\usepackage{pifont}

\newcommand{\bi}{\begin{itemize}}
\newcommand{\ei}{\end{itemize}}
\newcommand{\be}{\begin{enumerate}}
\newcommand{\ee}{\end{enumerate}}

\usepackage{makecell}
\usepackage[linesnumbered,ruled,vlined]{algorithm2e}
\usepackage{algorithmic}
\usepackage[skins]{tcolorbox}

\usepackage{amsthm}

\newtheorem{theorem}{Theorem}[section]
\newtheorem{proposition}[theorem]{Proposition}

\newtheorem{definition}[theorem]{Definition}

\usepackage{wrapfig}

\usepackage[skins]{tcolorbox}

\usepackage{diagbox}

\begin{document}
%

\title{Comparative Separation: Evaluating Separation on Comparative Judgment Test Data}

\author{Xiaoyin Xi, Neeku Capak, Kate Stockwell,
        Zhe Yu,~\IEEEmembership{Member,~IEEE}
\IEEEcompsocitemizethanks{\IEEEcompsocthanksitem Xiaoyin Xi and Zhe Yu are with the Department
of Software Engineering, Rochester Institute of Technology.\protect\\
E-mail: xx4455, zxyvse@rit.edu
\IEEEcompsocthanksitem Neeku Capak is with University of California, Santa Barbara.\protect\\
E-mail: ncapak@ucsb.edu
\IEEEcompsocthanksitem Kate Stockwell is with University of Vermont.\protect\\
E-mail: Katherine.Stockwell@uvm.edu
}
}

\IEEEtitleabstractindextext{%
\begin{abstract}
This research seeks to benefit the software engineering society by proposing comparative separation, a novel group fairness notion to evaluate the fairness of machine learning software on comparative judgment test data. 
Fairness issues have attracted increasing attention since machine learning software is increasingly used for high-stakes and high-risk decisions. It is the responsibility of all software developers to make their software accountable by ensuring that the machine learning software do not perform differently on different sensitive groups -- satisfying the separation criterion. However, evaluation of separation (equalized odds in binary classification) requires (usually human annotated) ground truth labels for each test data point. This motivates our work on analyzing whether separation can be evaluated on comparative judgment test data. Instead of asking humans to provide the ratings or categorical labels on each test data point, comparative judgments are made between pairs of data points such as A is better than B or B requires more effort than A. According to the law of comparative judgment in Psychology, providing such comparative judgments yields a lower cognitive burden for humans than providing ratings or categorical labels. 

This work first defines the novel fairness notion comparative separation on comparative judgment test data, and the metrics to evaluate comparative separation. Then, both theoretically and empirically, we show that in binary classification problems, comparative separation is equivalent to separation. We also introduce null hypothesis tests to statistically evaluate whether separation or comparative separation is satisfied. Lastly, we analyze the number of test data points and test data pairs required to achieve the same level of statistical power in the evaluation of separation and comparative separation, respectively. This work is the first to explore fairness evaluation on comparative judgment test data. It shows the feasibility and the practical benefits of using comparative judgment test data for model evaluations.
\end{abstract}
}

\maketitle

\IEEEdisplaynontitleabstractindextext

%
\IEEEpeerreviewmaketitle

\IEEEraisesectionheading{\section{Introduction}\label{sec:introduction}}

Machine learning (ML) systems are increasingly used in high-stakes decision-making contexts such as healthcare~\cite{mortazavi2016analysis}, criminal justice~\cite{redmond2002data,Angwin2016}, hiring~\cite{amazon}, and credit scoring~\cite{GERMAN}. Numerous studies have shown that data-driven models can unintentionally encode or amplify historical inequities, producing disparate outcomes across sensitive groups. For example, natural language models inherit human-like 
semantic biases \cite{Caliskan2017}, commercial face recognition systems show unequal accuracy across demographic subpopulations \cite{Buolamwini2018}, and risk assessment tools such as COMPAS demonstrate substantial racial disparities \cite{Angwin2016}. These observations highlight the critical need for 
fairness-aware model evaluation.

Among group fairness criteria, \emph{separation} -- equivalent to equalized odds in binary classification -- has emerged as one of the most influential definitions~\cite{hardt2016equality}. 
Separation requires the prediction $C$ to be conditionally independent of sensitive attributes $A$ given the ground truth label $Y$. It avoids key limitations of demographic parity~\cite{Dwork2012}, which 
refuses perfect predictors under differing base rates, and of individual fairness~\cite{Dwork2012b}, which depends on task-specific similarity metrics that may be difficult to define in practice. Separation has been widely adopted in fair classification \cite{hardt2016equality, Agarwal2019} and more recently in fair 
regression~\cite{Berk2017, XiFairReweighing2024}. 

A central challenge in evaluating separation is that it requires access to 
ground-truth labels for each instance in the test data. However, in many domains such labels are costly, noisy, subjective, or pose high cognitive burden for human annotators~\cite{yannakakis2015ratings,zhang2020learning} —- particularly in settings involving subjective assessments, educational measurement, or software engineering estimation. A substantial body of research demonstrates that \emph{comparative judgments} -— pairwise comparisons such as 
``A is better than B'' —- are more reliable, more consistent, 
and cognitively easier for humans to provide than direct ratings or categorical labels. This is consistent with 
the law of comparative judgment~\cite{Thurstone1927}. Empirical studies 
in education and assessment show a high reliability of comparative judgment~\cite{Bramley2015, Verhavert2019}, while comparative judgments as training data has proven effective in preference learning and ranking applications~\cite{Furnkranz2010, Brinker2004}.

Recent work in software engineering further confirms the benefits of comparative 
judgments: comparative learning models achieve performance competitive with or superior to regression models while reducing the annotation effort required from developers~\cite{khan2025efficient}. This result raises a compelling question for fairness 
evaluation: \textbf{Can separation be evaluated using only comparative judgments, instead of direct ratings or classifications?} If so, fairness assessment could be made dramatically more scalable and feasible in domains where comparative judgments are easier to obtain than direct ratings or classifications.

To answer this question, we propose \emph{comparative separation}, the first fairness 
notion defined entirely on comparative judgment test data. Rather than relying on the
ground-truth labels $Y$ for individual instances, comparative separation evaluates 
fairness over comparative judgments made between pairs of data points, $Y_{ij} = \mathrm{sgn}(Y_i - Y_j)$. 
We define comparative separation formally, introduce associated evaluation metrics, and establish an exact theoretical equivalence between comparative separation and separation in binary classification settings. This result shows feasibility of evaluating separation on comparative judgment test data. Another advantage of comparative separation is that it naturally generalizes to regression labels without any modification. It provides an easy way to evaluate group fairness in regression problems while existing metrics for separation evaluation on regression data utilize conditional mutual information and rely on additional predictors to approximate the probability densities~\cite{steinberg2020fairness,XiFairReweighing2024}. 

We further develop hypothesis-testing procedures  
for separation and comparative separation. Through power analysis on binary classification problems, we 
show that comparative separation achieves similar Type I/II error rates as 
separation with twice the number of pairwise evaluations —- an 
expected trade-off due to comparative data structure but justified by the reduced cognitive load of annotation \cite{Thurstone1927}. 

Our contributions are summarized as follows:

\begin{itemize}
    \item We introduce \emph{comparative separation}, a novel fairness notion that evaluates separation on comparative judgment test data.
    \item We provide formal metrics and statistical tests to evaluate comparative separation for both classification and regression problems.
    \item We prove that comparative separation is equivalent to 
    separation in binary classification.
    \item We present a statistical power analysis demonstrating that, in binary classifications, the evaluation of comparative separation requires twice the number of test data pairs than the evaluation of separation to reach the same level of statistical power.
    \item We empirically validate our findings through simulation and real-world fairness datasets in software engineering. All the code and data used in the experiments can be found publicly on GitHub\footnote{\url{https://github.com/hil-se/Comparative_Separation}}.
\end{itemize}

Overall, this work is the first to evaluate algorithmic fairness using comparative judgment test data instead of direct ratings or classifications. This significantly broadens the applicability of fairness evaluation to domains where comparative annotations are natural, reliable, or more cost-effective than 
direct ratings or classifications.

\section{Background}\label{sec:background}

\subsection{Group fairness and the Separation Criterion}

Machine learning models increasingly support decisions in high-stakes domains such 
as criminal justice, hiring, and healthcare. Extensive evidence shows that these 
models can exhibit significant performance disparities across demographic groups—
including systematic stereotypes embedded in language models \cite{Caliskan2017}, 
unequal accuracy in commercial facial recognition \cite{Buolamwini2018}, and 
racially disparate error rates in risk assessment tools such as COMPAS 
\cite{Angwin2016}. Such concerns have motivated the development of formal 
fairness criteria that quantify unwanted dependencies between predictions and 
sensitive attributes.

Among group fairness notions, \emph{separation}, also known as \emph{equalized 
odds} in binary classifications, has become particularly influential. Separation requires the prediction 
$C$ to be conditionally independent of the sensitive attributes $A$ given the ground 
truth label $Y$ \cite{hardt2016equality}. 
\begin{definition}
Given a test set $S = \{(x_i\in\mathbb{R}^d,y_i\in\mathbb{R}^1,a_i\in\{0,1\})\}^{n}_{i=1}$ sampled from a distribution $\mathcal{D}$ over a domain $X\times  Y \times  A$ and a prediction model $f_{\theta}:{X}\rightarrow {Y}$, the model satisfies comparative separation if and only if
\begin{equation}\label{separation}
{C} \perp {A} | {Y}
\end{equation}
where $c_i = f_{\theta}(x_i)\in{C}$ is the prediction of a test data point.
\end{definition}
This ensures parity in true positive rates and false positive rates across groups, aligning 
with operational needs in risk-based decision making. While demographic parity 
\cite{Dwork2012} and individual fairness \cite{Dwork2012b} provide alternative 
perspectives, they face practical limitations such as mismatched base rates or the 
need for domain-specific similarity metrics.

Recent work extends separation beyond binary classification. Steinberg et al.~\cite{steinberg2020fairness} proposed a set of metrics to evaluate separation on regression problems. Xi and Yu’s 
FairReweighing \cite{XiFairReweighing2024} further advanced this direction by extending the mutual information-based metrics to scenarios with continuous 
sensitive attributes and introduced a density-based preprocessing algorithm to improve separation in model training. These directions highlight an increasing interest in generalizing 
separation to flexible outcome and attribute spaces. However, these approaches utilize additional predictors to approximate the probability densities in the calculation of the metrics, thus may introduce errors and biases in the evaluation of separation on regression problems. Narasimhan et al.~\cite{narasimhan2020pairwise} proposed to evaluate Pairwise Fairness instead for regression and ranking problems -- using the regression labels to generate pairwise labels and evaluate pairwise fairness on the pairwise labels. This bypasses the difficulty of estimating the probability densities in separation evaluations. 

A central limitation remains: separation evaluation requires accurate instance-level 
ground-truth labels, which are often costly or unreliable to obtain—particularly in 
domains involving subjective assessment or knowledge work.

\subsection{Comparative Judgment}

Comparative judgment offers an alternative supervision mechanism that 
mitigates many challenges of absolute labeling. Rather than asking annotators to 
assign numerical scores or classifications, it asks them to choose between two items (e.g. which user story requires more effort to implement). Thurstone’s law of comparative judgment \cite{Thurstone1927} 
established that humans are more reliable and consistent in pairwise comparisons 
than in direct ratings. Subsequent psychometric work confirms that comparative judgment produces 
highly reliable measurement scales \cite{Bramley2015, Verhavert2019} and reduces 
cognitive burden by avoiding arbitrary or ambiguous rating scales.

Comparative judgment’s reliability has motivated its adoption in domains such as educational assessment, content evaluation, and human perception modeling. Its favorable statistical properties suggest that comparative judgment could serve as an effective source of ground-truth supervision when absolute labels are difficult or noisy.

\subsection{Machine Learning with Comparative Judgments}

Comparative data align naturally with pairwise learning and preference modeling. 
F{\"u}rnkranz and H{\"u}llermeier \cite{Furnkranz2010} formalized preference learning 
as a supervised learning paradigm, and Brinker \cite{Brinker2004} developed active 
learning strategies for reducing annotation costs in pairwise ranking tasks. Burges 
\emph{et al.} \cite{Burges2005} introduced RankNet, demonstrating the scalability and 
effectiveness of gradient-based learning-to-rank algorithms built on pairwise 
comparisons.

In software engineering, comparative labels have recently received attention because 
developers often struggle with producing consistent absolute effort estimates. Qian 
\emph{et al.} \cite{Qian2015} studied adaptive pairwise preference elicitation, while 
Islam \emph{et al.} \cite{khan2025efficient} showed that comparative learning 
models can match or exceed the performance of regression models trained on 
absolute story-point labels, while reducing annotation effort and cognitive load. 
These results underscore the practical advantages of pairwise annotations in 
settings where absolute labels are difficult to obtain or unreliable.

\subsection{Fairness Evaluation Under Alternative Supervision}

While comparative judgment has proven effective for training machine learning models, 
its potential for fairness evaluation remains underexplored. Existing fairness metrics 
(separation, demographic parity, calibration, etc.) all assume independent instance-level outcomes. The FairReweighing framework \cite{XiFairReweighing2024} 
and fairness discussions highlight that fairness evaluation often fails in domains where direct labels are unavailable or too expensive to collect. These works implicitly point to the need for fairness 
frameworks that operate under more flexible or cognitively tractable supervision.

Broader fairness literature recognizes similar challenges. Researchers have noted that labels can themselves be biased, noisy, or subjective, which complicates the measurement of fairness \cite{Mitchell2018, Kleinberg2018}. Other studies explore fairness in ranking \cite{Singh2018, Zehlike2017} or recommendation systems, but 
these works focus on fair algorithmic behavior, not on fairness evaluation from 
comparative judgments.

To date, little work investigates whether fairness evaluation -- particularly for error-rate–based criteria such as separation -— can be constructed directly from comparative judgments. The combination of insights from FairReweighing 
\cite{XiFairReweighing2024}, preference learning \cite{Furnkranz2010}, comparative judgment research \cite{Thurstone1927}, and comparative learning in software engineering \cite{khan2025efficient} motivates the exploration of whether fairness criteria can be defined in ways compatible with comparative judgments. This gap in the literature establishes the need for formulating fairness notions that operate on pairwise data and do not rely on absolute labels.

\section{Methodology}\label{sec:methodology}

Inspired by the pairwise fairness notion of Narasimhan et al.~\cite{narasimhan2020pairwise}, we push it one-step further by directly collecting pairwise annotations (comparative judgments) instead of pointwise annotations. In addition, we propose comparative separation in this section and prove that it is equivalent to separation in binary classification problems in RQ1. We also prove that the proposed comparative separation criterion can be evaluated with the same pairwise fairness metrics proposed by Narasimhan et al.~\cite{narasimhan2020pairwise} on comparative judgment test set. In this way, comparative separation can be naturally evaluated for both classification and regression problems. We then propose to evaluate comparative separation and separation with statistical tests in RQ2 and analyze the number of test data required to statistically evaluate comparative separation compared to that for separation in RQ3.

\subsection{Comparative Separation}

Similarly to the separation criterion $C\perp A | Y$, we define comparative separation as follows:
\begin{definition}
Consider a test set of comparative judgments $S_p = \{(x_i, x_j\in\mathbb{R}^d,y_{ij}\in\{-1,1\},a_i,a_j\in\{0,1\})\}^{n_p}$ sampled from a distribution $\mathcal{D}_{ij}$ over a domain $X_{ij}\times  Y_{ij} \times  A_{ij}$, where $x_{ij} = (x_i, x_j)\in X_{ij}$, $a_{ij} = (a_i, a_j)\in A_{ij}$, and $y_{ij} =\text{sgn}(y_i-y_j)\in  Y_{ij}$ is the comparative judgment between $x_i$ and $x_j$. Given a prediction model $f_{\theta}:{X}\rightarrow {Y}$, the model satisfies comparative separation if and only if
\begin{equation}\label{csp}
{C}_{ij} \perp {A}_{ij} | {Y}_{ij}
\end{equation}
where $c_{ij} = \text{sgn}(f_{\theta}({x}_i)-f_{\theta}(x_j)) \in {C}_{ij}$ is the comparative prediction of a test data pair.
\end{definition}

To evaluate the comparative separation in \eqref{csp}, we further define the comparative true positive rate $TPR(a_{ij})$:
\begin{equation}\label{metrics}
\begin{aligned}
&TPR(a_{ij}) = P(C_{ij} = 1| A_{ij} = a_{ij},\, Y_{ij} = 1),
\end{aligned}
\end{equation}
where $a_{ij} \in \{(0,0),(0,1),(1,0),(1,1)\}.$
\begin{theorem}\label{theorem1}
Comparative separation $C_{ij} \perp A_{ij} | Y_{ij}$ is satisfied if and only if \begin{equation}\label{condition}
\begin{aligned}
\forall a_{ij},\quad TPR(a_{ij}) =  \mathrm{c},
\end{aligned}
\end{equation}
where $\mathrm{c}$ is a constant.
\end{theorem}
\begin{proof}
We have $C_{ij} \perp A_{ij} | Y_{ij}$ if and only if 
\begin{equation}\label{eq:comp}
\begin{aligned}
&P(C_{ij} = c_{ij}| A_{ij} = a_{ij},\, Y_{ij} = y_{ij}) = P(C_{ij} = c_{ij}| Y_{ij} = y_{ij}) \\
&\forall c_{ij},\,y_{ij}\in\{-1,\,1\},\,a_{ij}\in\{(0,0),(0,1),(1,0),(1,1)\}.
\end{aligned}
\end{equation}
This is equivalent to 
\begin{equation}\label{eq:two}
\begin{aligned}
&P(C_{ij} = 1| A_{ij} = a_{ij},\, Y_{ij} = 1) = P(C_{ij} = 1| Y_{ij} = 1)\\
&\forall a_{ij}\in\{(0,0),(0,1),(1,0),(1,1)\},
\end{aligned}
\end{equation}
since \eqref{eq:sym} and \eqref{eq:sum}.
\begin{equation}\label{eq:sym}
\begin{aligned}
&P(C_{ij} = c_{ij}| A_{ij} = a_{ij},\, Y_{ij} = 1) \\= &P(C_{ij} = -c_{ij}| A_{ij} = a_{ji},\, Y_{ij} = -1)
\end{aligned}
\end{equation}
\begin{equation}\label{eq:sum}
\begin{aligned}
&P(C_{ij} = 1| A_{ij} = a_{ij},\, Y_{ij} = 1) +\\
&P(C_{ij} = -1| A_{ij} = a_{ij},\, Y_{ij} = 1) = 1.
\end{aligned}
\end{equation}
Apply \eqref{metrics} to \eqref{eq:two}, we have \eqref{condition} with
$$\mathrm{c}=P(C_{ij}=1|Y_{ij}=1).$$
\end{proof}

\subsection{RQ1: Is comparative separation equivalent to separation in binary classifications?} As a special case, comparative separation is equivalent to separation in binary classification.

\begin{theorem}\label{theorem2}
In binary classification problems where $Y,\, C\in\{0,\,1\}$, $C_{ij} \perp A_{ij} | Y_{ij} \Leftrightarrow C \perp A | Y$. 
\end{theorem}

\begin{proof}
When $Y,\, C\in\{0,\,1\}$, we have \eqref{eq:TPRij}.
\begin{equation}\label{eq:TPRij}
\begin{aligned}
&TPR(a_{ij}) = P(C_{ij} = 1| A_{ij} = a_{ij},\, Y_{ij} =1) \\
= &P(C_i = 1|A_{i} = a_{i},\,Y_i=1)\cdot P(C_j = 0|A_{j} = a_j,\,Y_j=0)\\
=&TPR(A=a_i)\cdot TNR(A=a_j),
\end{aligned}
\end{equation}
Now given $C_{ij} \perp A_{ij} | Y_{ij}$ and Theorem~\ref{theorem1}, we have
\begin{equation*}
\begin{aligned}
&TPR(A=1)\cdot TNR(A=1) \\= &TPR(A=0)\cdot TNR(A=1) \\
 = &TPR(A=1)\cdot TNR(A=0).
\end{aligned}
\end{equation*}
That is
\begin{equation}\label{eq:rate}
\begin{aligned}
&TPR(A=1) = TPR(A=0) \\
&TNR(A=1) = TNR(A=0).
\end{aligned}
\end{equation}
Therefore, we have $C_{ij} \perp A_{ij} | Y_{ij} \Rightarrow C\perp A|Y$.

Now given $C\perp A|Y$, we have \eqref{eq:rate}, thus
\begin{equation*}
\begin{aligned}
&TPR(A=a_i)\cdot TNR(A=a_j) = TPR(a_{ij}) = \mathrm{c}\\
&\forall a_i,\, a_j \in \{0,\,1\}. 
\end{aligned}
\end{equation*}
So we also have $$C\perp A|Y \Rightarrow C_{ij} \perp A_{ij} | Y_{ij}.$$
\end{proof}

\noindent\textbf{Answer to RQ1: }Theoretically, we have proved that comparative separation is equivalent to separation in binary classifications in Theorem~\ref{theorem2}. We will also validate this conclusion with simulations and experiments on real world data in Section~\ref{sec:Experiment}.

\subsection{RQ2: How to statistically test whether separation or comparative separation is satisfied for binary classifiers?}

\subsubsection{Separation}\label{sect:separation}
Traditionally, separation is tested simply by reporting the differences between the estimated metrics:
\begin{equation}\label{eq:aodeod}
\begin{aligned}
&EOD = \hat{TPR}(A=1) - \hat{TPR}(A=0) \\
&AOD = \frac{1}{2}(EOD + \hat{FPR}(A=1) - \hat{FPR}(A=0))
\end{aligned}
\end{equation}
where
\begin{equation}\label{eq:hattpr}
\begin{aligned}
&\hat{TPR}(A=a) =  \frac{|\{C=1, Y=1, A=a\}|}{|\{Y=1, A=a\}|}\\
&\hat{FPR}(A=a) =  \frac{|\{C=1, Y=0, A=a\}|}{|\{Y=0, A=a\}|}\\
\end{aligned}
\end{equation}
are the estimations of $TPR(A=a)$ and $FPR(A=a)$ with the test set and $|\{ \text{conditions}\}|$ represents the number of test data points satisfying the conditions. Given a test set with a finite number of $n$ data points, the estimations $\hat{TPR}(A=a)$ and $\hat{FPR}(A=a)$ have errors. To be more specific, based on the law of large numbers and the central limit theorem, these estimations approximate normal distributions with a large enough $n$. 
\begin{equation}\label{norm1}
\begin{aligned}
\hat{TPR}(A=a) = \overline{\{C|Y=1, A=a\}} \sim {N}(\mu, s^2)
\end{aligned}
\end{equation}
where $ \overline{\{C|Y=1, A=a\}}$ calculates the mean of $C$ across the test set samples with $Y=1, A=a$. $\mu=TPR(A=a)$ is the underlying real true positive rate (not the sampled one) and 
\begin{equation}\label{var_sep}
\begin{aligned}
s^2 = (1-\mu)\cdot\mu/|\{Y=1, A=a\}|
\end{aligned}
\end{equation}
is the sampled variance. A larger sample size $|\{Y=1, A=a\}|$ leads to a smaller variance and a more accurate estimation of $TPR(A=a)$. 

Now, in addition to simply reporting $EOD = \hat{TPR}(A=1) - \hat{TPR}(A=0)$, we can statistically analyze whether $TPR(A=1)$ and $TPR(A=0)$ are significantly different by testing the null hypothesis
$$H_{0}^{t}: TPR(A=1) = TPR(A=0).$$
When the sample size is large enough, $|\{Y=1\}|\ge 30$, a z-test can be utilized for the null hypothesis test of $H_{0}^{t}$:
\begin{equation*}
\begin{aligned}
z = \frac{\hat{TPR}(A=1) - \hat{TPR}(A=0)}{\sqrt{\hat{s}^2(A=1)+\hat{s}^2(A=0)}} 
\end{aligned}
\end{equation*}
Here, the sampled variances can be estimated as
\begin{equation}\label{eq_s2hat}
\begin{aligned}
\hat{s}^2 =\frac{ (1-\hat{TPR}(A=a))\cdot \hat{TPR}(A=a)}{|\{Y=1, A=a\}|}.
\end{aligned}
\end{equation}
When $|z|$ is large enough to reject the null hypothesis (e.g. with p value $< \alpha = 0.05$), we can conclude that separation is violated. The same can be applied to the estimation and test the second null hypothesis
$$H_{0}^{f}: FPR(A=1) = FPR(A=0).$$
Note that, separation is considered violated when either null hypothesis $H_{0}^{t}$ or $H_{0}^{f}$ is rejected. Therefore, the Type I error rate (false positive rate) of the statistical test of separation will be $1-(1-\alpha)^2 = 0.0975$. That is, when $C\perp A| Y$, there is a probability of $0.0975$ that at least one null hypothesis is rejected and thus separation is evaluated as violated.

\subsubsection{Comparative Separation}\label{sect:comp_separation}
Now, let us look at the comparative separation. Similar to separation, we can test whether comparative separation is satisfied with null hypotheses of $TPR(a_{ij}) = TPR(a_{pq})$. Given \eqref{eq:sym}, we have: 
\begin{equation}\label{TPR_hat}
\begin{aligned}
\hat{TPR}(a_{ij}) = \frac{\hat{TP}(a_{ij})+\hat{TN}(a_{ij})}{\hat{P}(a_{ij})+\hat{N}(a_{ij})}
\end{aligned}
\end{equation}
where
\begin{equation}
\begin{aligned}
\hat{TP}(a_{ij}) = &|\{C_{ij}=1, Y_{ij}=1, A_{ij}=a_{ij}\}|\\
\hat{TN}(a_{ij}) = &|\{C_{ij}=-1, Y_{ij}=-1, A_{ij}=a_{ji}\}|\\
\hat{P}(a_{ij}) = &|\{Y_{ij}=1, A_{ij}=a_{ij}\}|\\
\hat{N}(a_{ij}) = &|\{Y_{ij}=-1, A_{ij}=a_{ji}\}|\\
\end{aligned}
\end{equation}
Based on the law of large numbers and the central limit theorem, $\hat{TPR}(a_{ij})$ also approximates a normal distribution with a large enough sample size $\hat{P}(a_{ij})+\hat{N}(a_{ij})$.
\begin{equation}\label{norm2}
\begin{aligned}
\hat{TPR}(a_{ij}) = \overline{\{C_{ij}\odot Y_{ij}|Y_{ij}=1, A_{ij}=a_{ij}\}} \sim {N}(\mu_{ij}, s^2_{ij})
\end{aligned}
\end{equation}
where 
\begin{equation*}
C_{ij}\odot Y_{ij} = \left\{\begin{matrix} &1  &\text{when }C_{ij} = Y_{ij} \\ &0 & \text{when }C_{ij} \neq Y_{ij}
\end{matrix}\right..
\end{equation*}
$\mu_{ij}=TPR(a_{ij}) $ is the underlying real true positive rate (not the sampled one) and 
\begin{equation}\label{varij_sep}
\begin{aligned}
s^2_{ij} = (1-\mu_{ij})\cdot \mu_{ij}/(\hat{P}(a_{ij})+\hat{N}(a_{ij}))
\end{aligned}
\end{equation}
is the sampled variance. A larger sample size leads to a smaller variance and a more accurate estimation of $TPR(a_{ij})$. 

When the sample size is large enough, $\hat{P}(a_{ij})+\hat{N}(a_{ij})\ge 30$, a z-test can be utilized for the Null-hypothesis test of $TPR(a_{ij}) = TPR(a_{pq})$:
\begin{equation*}
\begin{aligned}
z = \frac{\hat{TPR}(a_{ij}) - \hat{TPR}(a_{pq})}{\sqrt{\hat{s}^2_{ij}+\hat{s}^2_{pq}}} 
\end{aligned}
\end{equation*}
Here, the sampled variances are estimated as
\begin{equation}\label{eq_s2ijhat}
\begin{aligned}
\hat{s}^2_{ij} = (1-\hat{TPR}(a_{ij}))\cdot \hat{TPR}(a_{ij})/(\hat{P}(a_{ij})+\hat{N}(a_{ij})).
\end{aligned}
\end{equation}
When $|z|$ is large enough to reject the null hypothesis (e.g. with p value $< \alpha = 0.05$), we can conclude that comparative separation is violated. To reach the same Type I error rate as separation, we can test the following two null hypotheses for binary classifiers:
\begin{equation}\label{h0c}
\begin{aligned}
H_0^{c}:TPR(1,0) = TPR(0,1),
\end{aligned}
\end{equation}
\begin{equation}\label{h0w}
\begin{aligned}
H_0^{w}:TPR(1,1) = TPR(0,0).
\end{aligned}
\end{equation}
\begin{proposition}
In binary classification problems where $Y,\, C\in\{0,\,1\}$, comparative separation $C_{ij} \perp A_{ij} | Y_{ij}$ is satisfied if and only if both null hypotheses $H_0^{c}$ and $H_0^{c}$ are accepted. 
\end{proposition}
\begin{proof}
When comparative separation is satisfied, it is straightforward to accept $H_0^{c}$ and $H_0^{c}$ given \eqref{condition}.
When both $H_0^{c}$ and $H_0^{c}$ are accepted, we have \eqref{h0c1} and \eqref{h0w1} given \eqref{eq:TPRij}.
\begin{equation}\label{h0c1}
\begin{aligned}
TPR(A=1)\cdot TNR(A=0) &= TPR(A=0)\cdot TNR(A=1),
\end{aligned}
\end{equation}
\begin{equation}\label{h0w1}
\begin{aligned}
TPR(A=1)\cdot TNR(A=1) &= TPR(A=0)\cdot TNR(A=0).
\end{aligned}
\end{equation}
Multiplying \eqref{h0c1} by \eqref{h0w1} we have
\begin{equation*}
\begin{aligned}
&TPR(A=1)^2\cdot TNR(A=0)\cdot TNR(A=1) \\
= &TPR(A=0)^2\cdot TNR(A=1)\cdot TNR(A=0),
\end{aligned}
\end{equation*}
therefore $TPR(A=1) = TPR(A=0)$. Similarly, we have $TNR(A=1)=TNR(A=0)$. thus separation is satisfied and comparative separation is also satisfied according to Theorem~\ref{theorem2}.
\end{proof}
\noindent The Type I error rate (false positive rate) of the statistical test of comparative separation is $1-(1-\alpha)^2 = 0.0975$ as well.

\noindent\textbf{Answer to RQ2: }Both separation and comparative separation can be statistically evaluated using two null hypotheses. The Type I error rates for these evaluations are decided by the value of $\alpha$. This will be validated with simulations later in Section~\ref{sec:Experiment}.

\subsection{RQ3: How many test data are sufficient to reach a desired statistical power of testing separation and comparative separation, respectively?}

As described in RQ2, the separation evaluation requires statistical tests of two null hypotheses $H_0^t$ and $H_0^f$. Similarly, the comparative separation evaluation requires statistical tests of two null hypotheses $H_0^c$ and $H_0^w$. The Type I error rate is $0.0975$ when $\alpha=0.05$ for both separation and comparative separation evaluations. Now, we analyze the statistical power (true positive rate, which is one minus the Type II error) of each test.

When testing null hypotheses, statistical power is affected by multiple factors. For example, the power of testing $H_0^t: TPR(A=1) = TPR(A=0)$ is affected by 
\be
\item
The statistical significance criterion $\alpha$. A larger $\alpha$ makes it easier to reject the null hypothesis and thus increases the statistical power (true positive rate) but also increases the Type I error rate (false positive rate), vice versa.
\item
The magnitude of the effect of interest. This magnitude is related to the actual difference $\Delta = TPR(A=1) - TPR(A=0)$ and the population variance $\sigma^2$ of $TPR(A=a)$. For a specific binary classifier, the magnitude of the effect (or the expected effect size) is a fixed value regardless of the test samples. 
\item
The size of the test data. More specifically, the number of test samples used to test $H_0^t$. A larger sample size will always increase the statistical power but requires more annotation effort.
\ee
With $\alpha = 0.05$ determined to ensure a low Type I error rate and the magnitude of effect being fixed based on the true $TPR$, the sample size of the test data is crucial to ensure a low Type II error rate. Fortunately, we can estimate the statistical power beforehand to see whether more test data is required.

\begin{proposition}\label{theorem_power}
Given two test sets $W = \{(w_i\in\{0,1\})\}^{n_w}_{i=1}$ and $V = \{(v_i\in\{0,1\})\}^{n_v}_{i=1}$, where $n_w, n_v\ge30$; and a null hypothesis $H_0: P(W=1) = P(V=1)$, the probability $\beta$ of accepting $H_0$ with a significance criterion $\alpha=0.05$ is
\begin{equation}\label{power}
\begin{aligned}
\beta = &F_{\overline{W}-\overline{V}}(1.96\sigma;\mu,\sigma) - F_{\overline{W}-\overline{V}}(-1.96\sigma;\mu,\sigma)
\end{aligned}
\end{equation}
where $F(t;\mu,\sigma)$ is the cumulative distribution function of $(\overline{W}-\overline{V}) \sim N(\mu, \sigma)$ taking on a value less than or equal to $t$, and  
\begin{equation}\label{power2}
\begin{aligned}
\mu =& P(W=1)-P(V=1)\\
\sigma = & \sqrt{\frac{\sigma_w^2}{n_w}+\frac{\sigma_v^2}{n_v}}\\
\sigma_w^2 =&  P(W=1)(1-P(W=1))\\
\sigma_v^2 =&  P(V=1)(1-P(V=1)).
\end{aligned}
\end{equation}
\end{proposition}
\begin{proof}
Based on central limit theorem and the law of large numbers, we have
\begin{equation}\label{wbar}
\begin{aligned}
\overline{W} \sim & N(P(W=1), \sqrt{\frac{\sigma_w^2}{n_w}})
\end{aligned}
\end{equation}
\begin{equation}\label{vbar}
\begin{aligned}
\overline{V} \sim & N(P(V=1), \sqrt{\frac{\sigma_v^2}{n_v}}).
\end{aligned}
\end{equation}
Subtract \eqref{vbar} from \eqref{wbar} we have
\begin{equation}\label{wv}
\begin{aligned}
(\overline{W} - \overline{V}) \sim & N(\mu, \sigma).
\end{aligned}
\end{equation}
\begin{figure}[tbh!]
    \centering
    \includegraphics[width=0.9\linewidth]{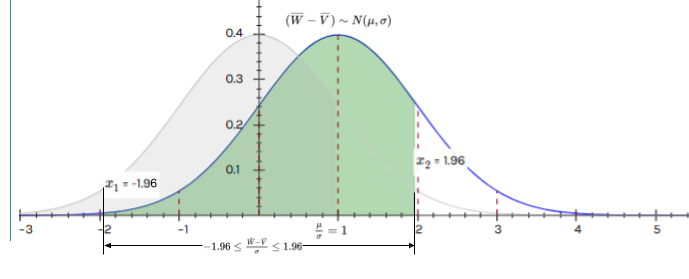}
    \caption{Illustration of the probability of $H_0$ being accepted with a two-tailed $p<\alpha=0.05$.}
    \label{fig:H0}
\end{figure}
As illustrated in Figure~\ref{fig:H0}, the probability of $H_0$ being accepted with a two-tailed $p<\alpha=0.05$ is the probability of $-1.96\sigma\le \overline{W} - \overline{V} \le 1.96\sigma$. Therefore we have
\begin{equation}\label{powerp}
\begin{aligned}
\beta = &P(-1.96\sigma\le \overline{W} - \overline{V} \le 1.96\sigma)\\
 = &F_{\overline{W}-\overline{V}}(1.96\sigma;\mu,\sigma) - F_{\overline{W}-\overline{V}}(-1.96\sigma;\mu,\sigma)
\end{aligned}
\end{equation}
\end{proof}

\subsubsection{Separation}

Separation is considered satisfied when both $H_0^t$ and $H_0^f$ are accepted. Therefore, the Type II error rate -- the probability of separation is tested as satisfied when it is actually violated -- on a test set $S$ is
\begin{equation}\label{beta}
\begin{aligned}
\beta_r = &\beta_t\cdot \beta_f.
\end{aligned}
\end{equation}
According to Proposition~\ref{theorem_power}, we have
\begin{equation}\label{betatf}
\begin{aligned}
\beta_t= &F_{\hat{TPR}(A=1)-\hat{TPR}(A=0)}(1.96\sigma_t;\mu_t,\sigma_t) - \\
&F_{\hat{TPR}(A=1)-\hat{TPR}(A=0)}(-1.96\sigma_t;\mu_t,\sigma_t)\\
\beta_f= &F_{\hat{FPR}(A=1)-\hat{FPR}(A=0)}(1.96\sigma_f;\mu_f,\sigma_f) - \\
&F_{\hat{FPR}(A=1)-\hat{FPR}(A=0)}(-1.96\sigma_f;\mu_f,\sigma_f).
\end{aligned}
\end{equation}
Following \eqref{power2}, 
\begin{equation}\label{tprtest}
\begin{aligned}
\mu_t = &TPR(A=1) - TPR(A=0)\\
\sigma_t = & \sqrt{\frac{\sigma_{t1}^2}{n_{t1}}+\frac{\sigma_{t0}^2}{n_{t0}}}\\
\sigma_{t1}^2 = & TPR(A=1)\cdot (1-TPR(A=1))\\
\sigma_{t0}^2 = & TPR(A=0)\cdot (1-TPR(A=0))\\
n_{t1} = & |\{Y=1, A=1\}|\\
n_{t0} = & |\{Y=1, A=0\}|.
\end{aligned}
\end{equation}
Given that $TPR(A=1)$ and $TPR(A=0)$ are predetermined by the classifier, the only factors affecting $\beta_t$ are the sample size $n_{t1}$ and $n_{t0}$. Similarly, the only factors affecting $\beta_t$ are the sample size $n_{f1} = |\{Y=0, A=1\}|$ and $n_{f0} = |\{Y=0, A=0\}|$. Assuming $Y, A\in\{0,1\}$ are uniformly distributed, we have $n_{t1} = n_{t0} = n_{f1}=n_{f0} = 0.25n$.

\subsubsection{Comparative Separation}

Comparative separation is considered satisfied when both $H_0^c$ and $H_0^w$ are accepted. Therefore, the Type II error rate -- the probability of comparative separation is tested as satisfied when it is actually violated --  on a test set $S_p$ is
\begin{equation}\label{betap}
\begin{aligned}
\beta_p = &\beta_c\cdot \beta_w.
\end{aligned}
\end{equation}
According to Proposition~\ref{theorem_power}, we have
\begin{equation}\label{betacw}
\begin{aligned}
\beta_c= &F_{\hat{TPR}(1,0)-\hat{TPR}(0,1)}(1.96\sigma_c;\mu_c,\sigma_c) - \\
&F_{\hat{TPR}(1,0)-\hat{TPR}(0,1)}(-1.96\sigma_c;\mu_c,\sigma_c)\\
\beta_w= &F_{\hat{TPR}(1,1)-\hat{TPR}(0,0)}(1.96\sigma_w;\mu_w,\sigma_w) - \\
&F_{\hat{TPR}(1,1)-\hat{TPR}(0,0)}(-1.96\sigma_w;\mu_w,\sigma_w).
\end{aligned}
\end{equation}
Following \eqref{power2}, 
\begin{equation}\label{tprc}
\begin{aligned}
\mu_c = &TPR(1,0) - TPR(0,1)\\
\sigma_c = & \sqrt{\frac{\sigma_{(1,0)}^2}{n_{(1,0)}}+\frac{\sigma_{(0,1)}^2}{n_{(0,1)}}}\\
\sigma_{(1,0)}^2 = & TPR(1,0)\cdot (1-TPR(1,0))\\
\sigma_{(0,1)}^2 = & TPR(0,1)\cdot (1-TPR(0,1))\\
n_{(1,0)} = & |\{Y_{ij}=1, A_{ij}=(1,0)\}|+|\{Y_{ij}=-1, A_{ij}=(0,1)\}|\\
n_{(0,1)} = &|\{Y_{ij}=1, A_{ij}=(0,1)\}|+|\{Y_{ij}=-1, A_{ij}=(1,0)\}|.
\end{aligned}
\end{equation}
Similarly, the only factors affecting $\beta_w$ are the sample size $n_{(1,1)} = |\{Y_{ij}=1, A_{ij}=(1,1)\}|+|\{Y_{ij}=-1, A_{ij}=(1,1)\}|$ and $n_{(0,0)} = |\{Y_{ij}=1, A_{ij}=(0,0)\}|+|\{Y_{ij}=-1, A_{ij}=(0,0)\}|$. Assuming $Y, A\in\{0,1\}$ are uniformly distributed, we have $n_{(1,0)} = n_{(0,1)} = n_{(1,1)}=n_{(0,0)} = \frac{1}{8}n_p$. This is because half of the sampled data with $y_i = y_j$ cannot be used to test comparative separation -- for example, a meaningful comparative judgment cannot be provided when a human rater is asked to compare which of the two dog images better resembles a dog. As a result, in the case of binary classification, two times the comparative judgments are needed in comparison to the test size of separation ($n_p = 2n$) to reach the same level of Type II error rates.

Note that, in regression problems where $P(Y_{ij}=0)\approx0$, we hypothesize that the same number of test data will be required to statistically evaluate comparative separation and separation ($n_p = n$). However, this requires further analysis and there does not exist any reliable metrics to evaluate separation on regression problems yet (the mutual information-based metrics require addition predictors to approximate the probabilities thus are not reliable).

\noindent\textbf{Answer to RQ3: }The Type II error rate of testing separation or comparative separation can be calculated when the true distribution of the test data is known. Roughly two times the comparative judgments are needed in comparison to the test size of separation ($n_p = 2n$) to reach the same level of Type II error rates in binary classifications. These conclusions will be validated with simulations and experiments on real world data in Section~\ref{sec:Experiment}.

\section{Simulations and Experiments}
\label{sec:Experiment}

In this section, the answers to the three research questions will be validated empirically through simulations and experiments on real world data.

\subsection{Simulations}

We simulate four different classifiers $f_{\theta_i}(x),\, i\in[0,3]$ on the same test data distribution. Table~\ref{tab:simulation} shows the underlying true distribution of the test data and the predictions of each classifier. For each classifier, $n=1,000$ or $2,000$ test data points are randomly sampled following this distribution to evaluate separation while $n_p=2n$ test data pairs are randomly sampled to evaluate comparative separation. Every evaluation is repeated for $r = 10,000$ times to measure the Type I or Type II error rate of the evaluation.

\begin{table}[!tbh]
\caption{Simulation distributions.}\label{tab:simulation}
\small
\centering
\setlength\tabcolsep{4pt}
\begin{tabular}{l|c|c|c|c|}
Underlying true probability & $f_{\theta_0}(x)$ & $f_{\theta_1}(x)$ & $f_{\theta_2}(x)$ & $f_{\theta_3}(x)$ \\\hline
$P(C = 1, Y = 1, A = 1)$      & 0.220         & 0.220         & 0.231        & 0.230         \\
$P(C = 0, Y = 1, A = 1)$      & 0.055        & 0.055        & 0.044        & 0.045        \\\hline
$P(C = 1, Y = 0, A = 1)$      & 0.090         & 0.081        & 0.081        & 0.105         \\
$P(C = 0, Y = 0, A = 1)$     & 0.135        & 0.144        & 0.144        & 0.120        \\\hline
$P(C = 1, Y = 1, A = 0)$      & 0.180         & 0.180         & 0.171        & 0.200         \\
$P(C = 0, Y = 1, A = 0)$      & 0.045        & 0.045        & 0.054        & 0.025        \\\hline
$P(C = 1, Y = 0, A = 0)$      & 0.110         & 0.121        & 0.121        & 0.100         \\
$P(C = 0, Y = 0, A = 0)$      & 0.165        & 0.154        & 0.154        & 0.175        \\\hline
\end{tabular}
\end{table}

\subsubsection{RQ1: Is comparative separation equivalent to separation in binary classifications?}

From the probabilities in Table~\ref{tab:simulation}, the true values of the metrics of separation and comparative separation can be calculated for each classifier. For example, 
\begin{equation*}
\begin{aligned}
TPR(A = a_i) &= P(C = 1| Y = 1, A = a_{i}) \\
&= \frac{P(C = 1, Y = 1, A = a_{i})}{P(Y = 1, A = a_{i})},
\end{aligned}
\end{equation*}
\begin{equation*}
\begin{aligned}
&TPR(a_{ij}) = P(C_{ij} = 1| Y_{ij} = 1, A_{ij} = a_{ij})\\
&=\frac{P(C=1, Y=1, A=a_i)\cdot P(C=0, Y=0, A=a_j)}{P(Y=1, A=a_i)\cdot P(Y=0, A=a_j)},
\end{aligned}
\end{equation*}
where $P(Y=y_i, A=a_{i}) = P(C = 1, Y = y_i, A = a_{i})+P(C = 0, Y = y_i, A = a_{i})$.

Table~\ref{tab:metrics} shows the metrics calculated for these four classifiers. As we can see, only $f_{\theta_0}$ satisfies separation and comparative separation. The other three classifiers violate both separation and comparative separation. This is consistent with Theorem~\ref{theorem2} by showing that separation and comparative separation are always satisfied or violated at the same time in binary classifications.

\begin{table}[!tbh]
\caption{Calculated metrics of separation and comparative separation.}\label{tab:metrics}
\small
\centering
\setlength\tabcolsep{4pt}
\begin{tabular}{l|c|c|c|c|}
Metric & $f_{\theta_0}(x)$ & $f_{\theta_1}(x)$ & $f_{\theta_2}(x)$ & $f_{\theta_3}(x)$ \\\hline
$TPR(A=1) - TPR(A=0)$      & 0         & 0         & 0.080        & -0.053         \\
$FPR(A=1) - FPR(A=0)$      & 0        & -0.080        & -0.080       & 0.103        \\\hline
$TPR(1,0) - TPR(0,1)$      & 0         & -0.064        & -0.016        & 0.058        \\
$TPR(1,1) - TPR(0,0)$     & 0        & 0.064        & 0.112        & -0.120        \\\hline
\end{tabular}
\end{table}

\subsubsection{RQ2: How to statistically test whether separation or comparative separation is satisfied for binary classifiers?}

The statistical tests designed for separation and comparative separation rely on the law of large numbers and the central limit theorem. More specifically, they rely on the fact that all the calculated metrics approximate normal distributions in large samples as described in \eqref{norm1} and \eqref{norm2}:
\begin{equation*}
\begin{aligned}
\hat{TPR}(A=a) &= \overline{\{C|Y=1, A=a\}} \sim {N}(\mu, s^2),\\
\hat{TPR}(a_{ij}) &= \overline{\{C_{ij}\odot Y_{ij}|Y_{ij}=1, A_{ij}=a_{ij}\}} \sim {N}(\mu_{ij}, s^2_{ij}).
\end{aligned}
\end{equation*}
Without loss of generality, we validate these by simulating $\hat{TPR}(A=1)$ and $\hat{TPR}(1,0)$ for the $f_{\theta_0}(x)$ classifier.

Based on the law of large numbers and the central limit theorem, the expected mean and variance of $\hat{TPR}(A=1)$ are:
\begin{equation*}
\begin{aligned}
\mu &= TPR(A=1) =  \frac{P(C = 1, Y = 1, A = 1)}{P(Y = 1, A = 1)} = 0.8,\\
s^2 &= \frac{(1-\mu)\cdot\mu}{|\{Y=1, A=1\}|} = \frac{(1-\mu)\cdot\mu}{P(Y = 1, A = 1)\cdot n} = \frac{0.58}{n}.
\end{aligned}
\end{equation*}
The mean and variance of $\hat{TPR}(A=1)$ of $n=1,000$ random samples over $r=10,000$ randomly repeated simulations are:
\begin{equation*}
\begin{aligned}
\overline{\hat{TPR}(A=1)} &=  0.79992 \approx \mu,\\
s^2(\hat{TPR}(A=1)) &=  0.00058 \approx s^2.
\end{aligned}
\end{equation*}
The mean and variance of $\hat{TPR}(A=1)$ of $n=2,000$ random samples over $r=10,000$ randomly repeated simulations are:
\begin{equation*}
\begin{aligned}
\overline{\hat{TPR}(A=1)} &=  0.80026 \approx \mu,\\
s^2(\hat{TPR}(A=1)) &=  0.00029 \approx s^2.
\end{aligned}
\end{equation*}

Similarly, the expected mean and variance of $\hat{TPR}(1,0)$ based on the law of large numbers and the central limit theorem are:
\begin{equation*}
\begin{aligned}
\mu_{ij} &= \hat{TPR}(1,0) =  0.48,\\
s_{ij}^2 &= \frac{(1-\mu_{ij})\cdot\mu_{ij}}{n_{(1,0)}} \\
&= \frac{(1-\mu_{ij})\cdot\mu_{ij}}{2P(Y = 1, A = 1)\cdot P(Y = 0, A = 0)\cdot n_p} \\
&= \frac{1.65}{n_p}.
\end{aligned}
\end{equation*}
The mean and variance of $\hat{TPR}(1,0)$ of $n_p=2,000$ random samples over $r=10,000$ repeated simulations are:
\begin{equation*}
\begin{aligned}
\overline{\hat{TPR}(1,0)} &=  0.47964 \approx \mu_{ij},\\
s^2(\hat{TPR}(1,0)) &=  0.00083 \approx s_{ij}^2.
\end{aligned}
\end{equation*}
The mean and variance of $\hat{TPR}(1,0)$ of $n_p=4,000$ random samples over $r=10,000$ repeated simulations are:
\begin{equation*}
\begin{aligned}
\overline{\hat{TPR}(1,0)} &=  0.48029 \approx \mu_{ij},\\
s^2(\hat{TPR}(1,0)) &=  0.00041 \approx s_{ij}^2.
\end{aligned}
\end{equation*}
The simulation results are consistent with the expected results. We can then confirm that the separation and comparative separation metrics follow normal distributions with the expected means and variances in large samples. In addition, the sampled variance of each metric is inversely proportional to the sample size of the test data.

\begin{table}[!tbh]
\caption{Probability of detecting the violation of separation and comparative separation $(1-\beta)$.}\label{tab:detection}
\small
\centering
\setlength\tabcolsep{2pt}
\begin{tabular}{l|l|c|c|c|c|}
\multicolumn{2}{l|}{\multirow{2}{*}{Classifier}}  & \multicolumn{2}{c|}{Separation}   & \multicolumn{2}{c|}{Comparative separation}  \\\cline{3-6}
\multicolumn{2}{l|}{} & $n$=1,000 & $n$=2,000 & $n_p$ = 2,000 & $n_p$=4,000 \\\hline
 \multirow{2}{*}{$f_{\theta_0}(x)$} &  Expected & 0.0975         & 0.0975   & 0.0975         & 0.0975     \\
  &  Simulation &  0.0960         &  0.0961 &    0.0977       &  0.0965   \\\hline
 \multirow{2}{*}{$f_{\theta_1}(x)$} &  Expected & 0.4743         & 0.7464 &  0.5032         & 0.7692    \\
  &  Simulation & 0.4772         &   0.7349   & 0.5042 & 0.7776  \\\hline
   \multirow{2}{*}{$f_{\theta_1}(x)$} &  Expected & 0.7800         &  0.9682  & 0.7274 & 0.9484   \\
  &  Simulation &  0.7856         &  0.9683  &  0.7254&   0.9447 \\\hline
   \multirow{2}{*}{$f_{\theta_1}(x)$} &  Expected &  0.7890         &  0.9712   &   0.8232& 0.9813  \\
  &  Simulation & 0.7879         &   0.9697  & 0.8210 & 0.9811  \\\hline
\end{tabular}
\end{table}

\subsubsection{RQ3: How many test data are sufficient to reach a desired statistical power of testing separation and comparative separation, respectively?}

Following Proposition~\ref{theorem_power}, the probability of accepting a null hypothesis can be calculated when the true probability distribution is known. Therefore, we can calculate the expected probabilities of separation and comparative separation to be evaluated as violated for each classifier in the simulations. For $f_{\theta_0}(x)$, this is the Type I error rate and is only related to $\alpha$ as discussed in RQ2:
$$1-\beta = 1-(1-\alpha)^2 = 0.0975.$$ 
For the other three classifiers, this is the statistical power $1-\beta_r$ for separation and $1-\beta_p$ for comparative separation. Then, we count the frequency of separation and comparative separation to be tested as violated over $r=10,000$ randomly repeated samples of $S$ and $S_p$ with $n=1,000$ and $n_p=2,000$, respectively. Table~\ref{tab:detection} shows the result of the expected $1-\beta$ calculated with the true probability distributions and the simulation results. We can observe from Table~\ref{tab:detection}:
\bi
\item
Proposition~\ref{theorem_power} is validated since all the simulation results are close to the expected values.
\item
The Type I error rate of evaluation both separation and comparative separation are only related to $\alpha$ as discussed in RQ2.
\item
The conclusion in RQ3 that two times the comparative judgments are needed in comparison to the test size of separation ($n_p = 2n$) to reach the same level of Type II error rate is also validated -- the Type II error rates of evaluating separation and comparative separation are similar when $n_p=2n$.
\item
For both separation and comparative separation evaluation, the type II error rates could always be reduced to a target value by increasing the size of the test data.
\ei

\subsection{Experiments on real world data}
Simulations have validated the theoretical analysis of comparative separation. 
In this subsection, we demonstrate the effectiveness of comparative separation with the application to three real world datasets. Two of them are binary classification datasets and one of them is a regression dataset. 

\begin{table}[!tbh]
\centering
\caption{Description of the binary classification datasets used in the experiment.}
\label{tab:class_data}
\setlength\tabcolsep{2pt}
\begin{tabular}{|l|c|c|c|c|c|}
\hline
\multirow{2}{*}{\textbf{Dataset}}    & \multirow{2}{*}{\textbf{$|X|$}} & \multicolumn{2}{c|}{\textbf{Sensitive Attribute}}                                                                                        & \multicolumn{2}{c|}{\textbf{Class Label}}                      \\ \cline{3-6}
 &            & \textbf{$A=1$}                     & \textbf{$A=0$}                                                   & \textbf{$Y=1$} & \textbf{$Y=0$} \\ \hline
\multirow{2}{*}{Compas}             & \multirow{2}{*}{7,214}                       & Sex-Male & Sex-Female & Did not                           &  \multirow{2}{*}{reoffended}      \\
& & Race-Caucasian & Race-Other & Reoffended &
\\ \hline
\multirow{2}{*}{German} & \multirow{2}{*}{1,000}                     & Sex-Male   & Sex-Female    & Good                             & Bad   \\
& & Age$>25$ & Age$\le 25$ & Credit & Credit \\
\hline
\end{tabular}
\end{table}

\subsubsection{Classification}
The two classification datasets used in this experiment are Compas~\cite{propublicadata} and German Credit Card~\cite{GERMAN} as shown in Table~\ref{tab:class_data}. These are two datasets commonly used in the analysis of algorithmic fairness. The Compas dataset simulates the scenario of a human judge making decisions on whether a defendant should be bailed out with ground truth labels indicating whether the defendant reoffended within two years. If comparative judgments are queried, judges can make easier decisions comparing which defendant of the two has a higher chance to reoffend. The German Credit Card dataset simulates the scenario of a credit card agent making decisions on whether an application should be approved. Similarly, comparative judgments on which of the two applicants has better credit are easier for the agent to decide. We experiment with these two classification datasets to see whether comparative separation can be applied to evaluate biases on sensitive attributes such as age, race, and sex on a test set of comparative judgments.

\begin{table}[!tbh]
\caption{Probability of detecting the violation of separation and comparative separation estimated by 1,000 random runs. Comparative separation is evaluated with $N_p = 2N$ randomly sampled pairs of comparative judgments.}\label{tab:class_result}
\small
\centering
\setlength\tabcolsep{2pt}
\begin{tabular}{l|l|cc|cc|}
\multirow{2}{*}{Treatment}   & \multirow{2}{*}{Metric}                 & \multicolumn{2}{c|}{Compas} &        \multicolumn{2}{c|}{German}        \\ \cline{3-6}
            &                        & sex    & race  & age    & sex   \\ \hline
\multirow{2}{*}{None}        & Separation             & 1      & 1     & 0.844  & 0.523 \\
            & Comparative Separation & 1      & 1     & 0.676  & 0.613 \\ \hline
\multirow{2}{*}{FairBalance} & Separation             & 0.245  & 0.37  & 0.241  & 0.18  \\
            & Comparative Separation & 0.312  & 0.411 & 0.323  & 0.313 \\ \hline
\multirow{2}{*}{Reweighing}  & Separation             & 0.286  & 0.642 & 0.223  & 0.135 \\
            & Comparative Separation & 0.38   & 0.455 & 0.325  & 0.35 \\\hline
\end{tabular}
\end{table}

For each experiment, the dataset is randomly split into 50\% training and 50\% testing. Then, one of the following three pre-processing treatments is applied to assign different sampling weights to the training data points to improve separation before training a logistic regression model with the default hyperparameters on scikit-learn\footnote{\url{https://scikit-learn.org/}}:
\bi
\item
\textbf{None}: no treatment is applied, the sample weight of every training data point is the same $w(x\in X) = 1$.
\item
\textbf{Reweighing~\cite{kamiran2012data}}: the sample weight of every training data point is determined by its sensitive group and class label $w_{RW}(a,y) = \frac{|A=a|\cdot|Y=y|}{|A=a, Y=y|}$. The goal is to make the weighted class distribution of every sensitive group the same.
\item
\textbf{FairBalance~\cite{yu2024fairbalance}}: the sample weight of every training data point is determined by its sensitive group and class label $w_{FB}(a, y) = \frac{|A=a|}{|A=a, Y=y|}$. The goal is to make the weighted class distribution of every sensitive group $1:1$.
\ei

\begin{table*}[!tbh]
\caption{Example backlog items in the Jira Software dataset.}
\small
\centering
\setlength\tabcolsep{3pt}
\label{tab:example_item}
\begin{tabular}{l|p{120mm}|c|c|}
issue key & title (description is omitted)    & is internal $A$ & story point $R$ \\\hline
JSW-1271 & As a JIRA Administrator I would like to be able to change the trigger of the night service     & 1            & 5          \\\hline
JSW-2478 & As a user I would like the ability to see a horizontal swimlane on my Rapid Board        & 1            & 13         \\\hline
JSW-3107 & As a user I would like the global statuses to be sorted by name in the Manage Rapid View to assist me when creating a new column & 1            & 2          \\\hline
JSW-1681 & Generic webwork aliases may clash with other plugins   & 0            & 5          \\\hline
JSW-2881 & Investigate: Log work operation of a task in the task board redirect the page to "ALL Assignee".  & 0            & 2          \\\hline
JSW-4768 & MVR: Make the burndown chart understand when an issue was added/removed from Sprint (i.e understand scope change)  & 0            & 13        \\\hline
\end{tabular}
\end{table*}

\noindent Table~\ref{tab:class_result} shows the probability of detecting violations of separation and comparative separation of the logistic regression models on the test set when the experiments are repeated for 1,000 times for each treatment and each dataset. From Table~\ref{tab:class_result}, we can observe that:
\be
\item
As expected, with training data pre-processed by FairBalance or Reweighing, the probabilities of the trained model violating separation or comparative separation become significantly lower (than None). This is consistent with the conclusions in \cite{yu2024fairbalance} and \cite{kamiran2012data}.
\item
The evaluation results of comparative separation are consistent with those of separation. This again validates our conclusions in Section~\ref{sec:methodology} that comparative separation is equivalent to separation for binary classification problems.
\ee

\subsubsection{Regression}

While the null-hypothesis tests for separation in Section~\ref{sect:separation} can only be applied to binary classification problems, the comparative separation tests in Section~\ref{sect:comp_separation} can be applied to both classification and regression problems. To test the effectiveness of comparative separation in regression problems, we used a dataset of story point estimation introduced by Choetkiertikul et al.~\cite{fu2022gpt2sp}. This dataset consists of the titles and descriptions of the backlog items, as well as their assigned story points collected through JIRA in 16 projects. During the pre-processing step, the titles and descriptions are concatenated as the final text feature~\cite{fu2022gpt2sp}. We specifically used the Jira Software project by separating project items into two sensitive groups: (1) internal user stories $A=0$ following the format ``As a [ROLE], I would like ...'' and (2) external defect-fixing requests $A=1$ which do not follow the user story format. Table~\ref{tab:example_item} shows three example backlog items from each sensitive group. The statistics of the Jira Software project is shown in Table~\ref{tab:jira_data}.

\begin{table}[!tbh]
\caption{Description of the Jira Software story point estimation dataset.}
\small
\centering
\setlength\tabcolsep{2pt}
\label{tab:jira_data}
\begin{tabular}{l|c|c|}
 \#Items  & Sensitive Attributes  & Story Points  \\\hline
\multirow{2}{*}{$|X|=284$}   &  $A=1$  internal: 127        & \multirow{2}{*}{$y \in \{1,2,3,5,8,13,20\}$}  \\
 & $A=0$   external: 157          &                    \\\hline
\end{tabular}
\end{table}

Story point estimation is an essential part of agile software development. Story points are unitless, project-specific estimates that help developers communicate their assumptions, discuss effort costs, and plan their upcoming sprint~\cite{pasuksmit2024systematic,tawosi2022investigating}. In every sprint, developers spend hours discussing and reaching a consensus on the estimated story points of each item on the backlog. Underestimating the items of one group over the other would cause severe and unfair consequences, such as overloading developers in charge of defect-fixing requests or overlooking their work. We experiment with this dataset to see whether comparative separation can be applied to evaluate biases for a regression model on a test set of comparative judgments.

Each time, the concatenated title and description of each backlog item is embedded by a pretrained Sentence-BERT model~\cite{reimers-2019-sentence-bert,reimers-2020-multilingual-sentence-bert} ``all-MiniLM-L6-v2'', then the embeddings are fed into a linear dense layer to predict the story points. The dense layer weights are trained with an Adam optimizer on the MAE loss. Before training, the following two pre-processing treatments are compared for their effects on mitigating algorithmic bias:
\bi
\item
\textbf{None}: no treatment is applied, the sample weight of every training data point is the same $w(x\in X) = 1$.
\item
\textbf{FairReweighing~\cite{XiFairReweighing2024}}: the sample weight of every training data point is determined by its sensitive group and ground truth story points: $w(a, y) = \frac{ P(a) P(y)}{P(a, y)}$ where $P(a)$, $P(y)$, and $P(a,y)$ are estimated probability densities. This is an adapted Reweighing algorithm for regression data.
\ei
Table~\ref{tab:regression_result} shows the experimental results. Similar to the reports from Xi and Yu~\cite{XiFairReweighing2024}, the (conditional) mutual information-based metric for separation $\hat{I}_{sep} = \frac{1}{n} \sum_{i=1}^n \log \frac{\rho(a_i\mid y_i, \hat{y}_i)}{\rho(a_i \mid y_i)}$ is significantly reduced after applying FairReweighing. However, the probability of violating comparative separation for FairReweighing is only slightly lower than that of None. This could indicate two possibilities:
\be
\item
The learned models of both None and FairReweighing treatments have little bias so that the bias mitigation treatment FairReweighing has little impact. This is supported by the low $\hat{I}_{sep}$ values for both treatments and the fact that their probabilities of violating comparative separation are similar to the ones of FairBalance and Reweighing in Table~\ref{tab:class_result}.
\item
It is also possible that comparative separation is not equivalent to separation when the dependent variables are continuous (regression data).
\ee

\begin{table}[!tbh]
\caption{Experimental results on the Jira Software story point estimation data. The Pearson and Spearman's rank correlation coefficient between the predicted story points and the ground truth story points are presented as $\rho$ and $r_s$. $\hat{I}_{sep}$ is a metric to evaluate the violation of separation~\cite{steinberg2020fairness}. Comparative separation is evaluated with $N_p = N$ randomly sampled pairs of comparative judgments. Probability of detecting the violation of comparative separation is estimated by 1,000 random runs. All the metrics are averaged from 10 random experiments.}\label{tab:regression_result}
\small
\centering
\setlength\tabcolsep{4pt}
\begin{tabular}{l|c|c|c|c|p{20mm}|}
Treatment      & MAE   & $\rho$ & $r_s$ & $\hat{I}_{sep}$  & Comparative Separation \\\hline
None           & 1.548 & 0.473   & 0.436    & 0.011 & 0.334                  \\
FairReweighing & 1.566 & 0.462   & 0.420    & 0.006 & 0.313 \\\hline
\end{tabular}
\end{table}



\section{Threats to validity}

\noindent \textbf{Conclusion validity: }As discussed in RQ2, the evaluation of comparative separation should take into consideration its Type I and Type II errors. The Type I error rate (the probability of finding a bias when there is none) is only related to $\alpha$ of the null hypothesis tests, while the Type II error rate (the probability of missing a bias when there exists one) is closely related to the sample size of the test set. Therefore, it is crucial to ensure a large enough test set when testing for both separation and comparative separation with the null hypotheses. 

\noindent \textbf{Construct validity: }The theoretical analysis in Section~\ref{sec:methodology} ensures that the designed tests directly measure the violations of separation or comparative separation. The simulation results and experiments on real world data in Section~\ref{sec:Experiment} further validate the theoretical analysis.

\noindent \textbf{Internal validity: }Sampling bias can have an impact on the evaluations of both separation and comparative separation. In the simulations and experiments, we minimize the sampling bias by repeating the experiments for 1,000 times each with random samples from the ground truth distribution or the dataset. 

\noindent \textbf{External validity: }This work suffers from external validity. 1. The theoretical analysis in Section~\ref{sec:methodology} only proves that comparative separation is equivalent to separation in binary classification. 2. The experimental result on the story point estimation dataset suggests that comparative separation may not be equivalent to separation in regressions. Therefore, whether comparative separation can be generalized to scenarios other than binary classification (such as regression, continuous sensitive attribute, etc.) requires further investigation. 



\section{Conclusion and future work}
\label{sec:Conclusions}
This paper proposes a novel fairness notion comparative separation to evaluate bias in machine learning software with comparative judgment test data. This alleviates the need of accurately annotated point-wise test data labels by only requiring pairwise comparisons between pairs of test data points -- e.g. A is better than B or user story B requires more effort than user story A. According to the law of comparative judgment~\cite{Thurstone1927}, providing such comparative judgments yields a lower cognitive burden for humans than deciding direct ratings or class labels. Empirical studies~\cite{Bramley2015, Verhavert2019} also find comparative judgments to have higher reliability and consistency than human annotated direct ratings or class labels. This work demonstrated the effectiveness of comparative separation with three research questions. In RQ1, we proved that comparative separation is equivalent to separation in binary classification problems. In RQ2, we showed the effectiveness of statistically testing separation and comparative separation with null hypothesis tests. In RQ3, we analyzed the relationship between the statistical power of the null hypothesis tests and the size of testing data. The theoretical analysis in the research questions were then supported by simulations on synthetic data and experiments on real world data.

To resolve the external validity threat that all the analysis were made for the binary classification scenario, we will explore the future work:
\bi
\item
Theoretically analyze the relationship between comparative separation and separation in regression problems.
\item
Experiment on more real world regression datasets to see whether the evaluation of comparation separation is consistent with separation.
\item
Conduct human subjects experiments to evaluate the effectiveness and efficiency of collecting comparative judgments as test data.
\ei
Overall, we believe that this work could benefit the software engineering research community by presenting a novel way of evaluating fairness and bias of machine learning software with comparative judgment test data. 

\section*{Acknowledgment}
This work is funded by NSF grants 2245796 and 2447631.

\bibliographystyle{IEEEtran}
\bibliography{zhe}

@MISC{amazon,
  title={Amazon scraps secret AI recruiting tool that showed bias against women},
  author= {Dastin, Jeffrey},
  howpublished = {https://www.reuters.com/article/us-amazon-com-jobs-automation-insight/amazon-scraps-secret-ai-recruiting-tool-that-showed-bias-against-women-idUSKCN1MK08G},
  publisher = {https://www.reuters.com/},
  year={2018}
}

@MISC{propublicadata,
  title={data for the propublica story 'machine bias'.},
  author= {propublica},
  howpublished = {\url{https://github.com/propublica/compas-analysis/}},
  year={2016}
}

@article{mortazavi2016analysis,
  title={Analysis of machine learning techniques for heart failure readmissions},
  author={Mortazavi, Bobak J and Downing, Nicholas S and Bucholz, Emily M and Dharmarajan, Kumar and Manhapra, Ajay and Li, Shu-Xia and Negahban, Sahand N and Krumholz, Harlan M},
  journal={Circulation: Cardiovascular Quality and Outcomes},
  volume={9},
  number={6},
  pages={629--640},
  year={2016},
  publisher={Am Heart Assoc}
}

@inproceedings{hardt2016equality,
  title={Equality of opportunity in supervised learning},
  author={Hardt, Moritz and Price, Eric and Srebro, Nati},
  booktitle={Advances in neural information processing systems},
  pages={3315--3323},
  year={2016}
}

@article{kamiran2012data,
  title={Data preprocessing techniques for classification without discrimination},
  author={Kamiran, Faisal and Calders, Toon},
  journal={Knowledge and Information Systems},
  volume={33},
  number={1},
  pages={1--33},
  year={2012},
  publisher={Springer}
}

@misc{COMPAS,
  title={propublica/compas-analysis},
  url = {https://github.com/propublica/compas-analysis},
  year={2015}
}

@article{GERMAN,
  title={UCI:Statlog (German Credit Data) Data Set},
  url = {https://archive.ics.uci.edu/ml/datasets/Statlog+(German+Credit
  +Data)},
  year={2000}
}

@article{yu2024fairbalance,
  title={FairBalance: How to Achieve Equalized Odds With Data Pre-processing},
  author={Yu, Zhe and Chakraborty, Joymallya and Menzies, Tim},
  journal={IEEE Transactions on Software Engineering},
  year={2024},
  publisher={IEEE}
}

@article{fu2022gpt2sp,
  title={GPT2SP: A transformer-based agile story point estimation approach},
  author={Fu, Michael and Tantithamthavorn, Chakkrit},
  journal={IEEE Transactions on Software Engineering},
  volume={49},
  number={2},
  pages={611--625},
  year={2022},
  publisher={IEEE}
}

@article{pasuksmit2024systematic,
  title={A Systematic Literature Review on Reasons and Approaches for Accurate Effort Estimations in Agile},
  author={Pasuksmit, Jirat and Thongtanunam, Patanamon and Karunasekera, Shanika},
  journal={ACM Computing Surveys},
  year={2024},
  publisher={ACM New York, NY}
}

@inproceedings{tawosi2022investigating,
  title={Investigating the effectiveness of clustering for story point estimation},
  author={Tawosi, Vali and Al-Subaihin, Afnan and Sarro, Federica},
  booktitle={2022 IEEE International Conference on Software Analysis, Evolution and Reengineering (SANER)},
  pages={827--838},
  year={2022},
  organization={IEEE}
}

@article{Caliskan2017,
  title={Semantics derived automatically from language corpora contain human-like biases},
  author={Caliskan, Aylin and Bryson, Joanna J. and Narayanan, Arvind},
  journal={Science},
  volume={356},
  number={6334},
  pages={183--186},
  year={2017}
}

@inproceedings{Buolamwini2018,
  title={Gender Shades: Intersectional Accuracy Disparities in Commercial Gender Classification},
  author={Buolamwini, Joy and Gebru, Timnit},
  booktitle={Proceedings of the Conference on Fairness, Accountability, and Transparency (FAT*)},
  pages={77--91},
  year={2018}
}

@misc{Angwin2016,
  title={Machine Bias: There’s software used across the country to predict future criminals. And it’s biased against blacks},
  author={Angwin, Julia and Larson, Jeff and Mattu, Surya and Kirchner, Lauren},
  howpublished={ProPublica},
  year={2016},
  note={\url{https://www.propublica.org/article/machine-bias-risk-assessments-in-criminal-sentencing}}
}

@article{redmond2002data,
  title={A data-driven software tool for enabling cooperative information sharing among police departments},
  author={Redmond, Michael and Baveja, Alok},
  journal={European Journal of Operational Research},
  volume={141},
  number={3},
  pages={660--678},
  year={2002},
  publisher={Elsevier}
}

@article{yannakakis2015ratings,
  title={Ratings are overrated!},
  author={Yannakakis, Georgios N and Mart{\'\i}nez, H{\'e}ctor P},
  journal={Frontiers in ICT},
  volume={2},
  pages={13},
  year={2015},
  publisher={Frontiers Media SA}
}

@article{zhang2020learning,
  title={Learning from crowdsourced ordinal data: A survey},
  author={Zhang, Jing and Wang, Xin and Wu, Xindong},
  journal={IEEE Transactions on Knowledge and Data Engineering},
  year={2020},
  publisher={IEEE}
}

@article{Dwork2012,
  title={Fairness Through Awareness},
  author={Dwork, Cynthia and Hardt, Moritz and Pitassi, Toniann and Reingold, Omer and Zemel, Richard},
  journal={Proceedings of Innovations in Theoretical Computer Science},
  pages={214--226},
  year={2012}
}

@article{Dwork2012b,
  title={The Geometry of Individual Fairness},
  author={Dwork, Cynthia and Ilvento, Christina},
  journal={Manuscript},
  year={2018}
}

@inproceedings{Agarwal2019,
  title={A Reductions Approach to Fair Classification},
  author={Agarwal, Alekh and Beygelzimer, Alina and Dud{\'i}k, Miroslav and Langford, John and Wallach, Hanna},
  booktitle={Proceedings of the International Conference on Machine Learning (ICML)},
  pages={60--69},
  year={2018}
}

@article{Berk2017,
  title={A Convex Framework for Fair Regression},
  author={Berk, Richard and Heidari, Hoda and Jabbari, Shahin and Joseph, Matthew and Kearns, Michael and Morgenstern, Jamie and Neel, Seth and Roth, Aaron},
  journal={arXiv preprint arXiv:1706.02409},
  year={2017}
}

@article{XiFairReweighing2024,
  title={FairReweighing: A Density-Based Preprocessing Framework for Fair Regression and Classification},
  author={Xi, Xiaoyin and Yu, Zhe},
  journal={(Manuscript / Preprint)},
  year={2024}
}

@article{Thurstone1927,
  title={A Law of Comparative Judgment},
  author={Thurstone, Louis L.},
  journal={Psychological Review},
  volume={34},
  number={4},
  pages={273--286},
  year={1927}
}

@article{Bramley2015,
  title={Investigating the Reliability of Adaptive Comparative Judgement},
  author={Bramley, Tom},
  journal={Cambridge Assessment Research Report},
  year={2015}
}

@article{Verhavert2019,
  title={A Meta-Analysis on the Reliability of Comparative Judgement},
  author={Verhavert, Samuel and Bouwer, Renske and Donche, Vincent and De Maeyer, Sven},
  journal={Assessment in Education: Principles, Policy \& Practice},
  volume={26},
  number={5},
  pages={541--562},
  year={2019}
}

@incollection{Furnkranz2010,
  title={Preference Learning},
  author={F{\"u}rnkranz, Johannes and H{\"u}llermeier, Eyke},
  booktitle={Preference Learning},
  publisher={Springer},
  pages={3--20},
  year={2010}
}

@inproceedings{Brinker2004,
  title={Active Learning of Label Ranking Functions},
  author={Brinker, Klaus},
  booktitle={Proceedings of the International Conference on Machine Learning (ICML)},
  year={2004}
}

@inproceedings{Burges2005,
  title={Learning to Rank Using Gradient Descent},
  author={Burges, Christopher and Shaked, Tal and Renshaw, Erin and Lazier, Ari and Deeds, Matthew and Hamilton, Nicole and Hullender, Greg},
  booktitle={Proceedings of the International Conference on Machine Learning (ICML)},
  pages={89--96},
  year={2005}
}

@article{Qian2015,
  title={Learning User Preferences by Adaptive Pairwise Comparison},
  author={Qian, Liu and Gao, Jing and Jagadish, H. V.},
  journal={Proceedings of the VLDB Endowment},
  volume={8},
  number={11},
  pages={1322--1333},
  year={2015}
}

@article{khan2025efficient,
  title={Efficient Story Point Estimation With Comparative Learning},
  author={Khan, Monoshiz Mahbub and Xi, Xioayin and Meneely, Andrew and Yu, Zhe},
  journal={arXiv preprint arXiv:2507.14642},
  year={2025}
}

@article{Mitchell2018,
  title={Prediction-Based Decisions and Fairness: A Catalogue of Choices, Assumptions, and Definitions},
  author={Mitchell, Shira and Potash, Eric and Barocas, Solon and D'Amour, Alexander and Lum, Kristian},
  journal={arXiv preprint arXiv:1811.07867},
  year={2018}
}

@article{Kleinberg2018,
  title={Inherent Trade-Offs in the Fair Determination of Risk Scores},
  author={Kleinberg, Jon and Mullainathan, Sendhil and Raghavan, Manish},
  journal={Proceedings of Innovations in Theoretical Computer Science (ITCS)},
  year={2017}
}

@inproceedings{Singh2018,
  title={Fairness of Exposure in Rankings},
  author={Singh, Ashudeep and Joachims, Thorsten},
  booktitle={Proceedings of the Conference on Fairness, Accountability, and Transparency (FAT*)},
  pages={93--102},
  year={2018}
}

@inproceedings{Zehlike2017,
  title={FA*IR: A Fair Top-$k$ Ranking Algorithm},
  author={Zehlike, Meike and Bonchi, Francesco and Castillo, Carlos},
  booktitle={Proceedings of the ACM Conference on Information and Knowledge Management (CIKM)},
  pages={1569--1578},
  year={2017}
}

@article{steinberg2020fairness,
  title={Fairness measures for regression via probabilistic classification},
  author={Steinberg, Daniel and Reid, Alistair and O'Callaghan, Simon},
  journal={arXiv preprint arXiv:2001.06089},
  year={2020}
}

@inproceedings{narasimhan2020pairwise,
  title={Pairwise fairness for ranking and regression},
  author={Narasimhan, Harikrishna and Cotter, Andrew and Gupta, Maya and Wang, Serena},
  booktitle={Proceedings of the AAAI Conference on Artificial Intelligence},
  volume={34},
  number={04},
  pages={5248--5255},
  year={2020}
}

@inproceedings{reimers-2019-sentence-bert,
  title = "Sentence-BERT: Sentence Embeddings using Siamese BERT-Networks",
  author = "Reimers, Nils and Gurevych, Iryna",
  booktitle = "Proceedings of the 2019 Conference on Empirical Methods in Natural Language Processing",
  month = "11",
  year = "2019",
  publisher = "Association for Computational Linguistics",
  url = "https://arxiv.org/abs/1908.10084",
}

@inproceedings{reimers-2020-multilingual-sentence-bert,
  title = "Making Monolingual Sentence Embeddings Multilingual using Knowledge Distillation",
  author = "Reimers, Nils and Gurevych, Iryna",
  booktitle = "Proceedings of the 2020 Conference on Empirical Methods in Natural Language Processing",
  month = "11",
  year = "2020",
  publisher = "Association for Computational Linguistics",
  url = "https://arxiv.org/abs/2004.09813",
}
\end{document}